\documentclass[conference]{IEEEtran}
\usepackage{graphicx}
\usepackage{amsmath}
\usepackage{amsfonts}
\usepackage{amssymb}

\newtheorem{theorem}{Theorem}

\newenvironment{proof}[1][Proof]{\noindent\textbf{#1.} }{\ \rule{0.5em}{0.5em}}

\begin{document}

\title{Optimal Power Control for Analog Bidirectional Relaying with
Long-Term Relay Power Constraint }
\author{%
\IEEEauthorblockN{Zoran Hadzi-Velkov\IEEEauthorrefmark{1}, Nikola
Zlatanov\IEEEauthorrefmark{2}, and Robert Schober\IEEEauthorrefmark{3}
\vspace{3mm} }
\IEEEauthorblockA{\IEEEauthorrefmark{1}Faculty of Electrical
Engineering and Information Technologies, Ss. Cyril and Methodius
University, Skopje, Macedonia \vspace{-0.0mm}}
\IEEEauthorblockA{\IEEEauthorrefmark{2}Department of Electrical
and Computer Engineering, University of British Columbia,
Vancouver, Canada
\IEEEauthorblockA{\IEEEauthorrefmark{3}Department of Electrical,
Electronics and Communications Engineering, University of
Erlangen-Nuremberg, Germany \vspace{-0.0mm}} \vspace{-5mm}}}
\maketitle

\let\thefootnote\relax\footnotetext{This work has been supported by the Alexander
von Humboldt fellowship program for experienced researchers}

\begin{abstract}
Wireless systems that carry delay-sensitive information (such as speech
and/or video signals) typically transmit with fixed data rates, but may
occasionally suffer from transmission outages caused by the random nature of
the fading channels. If the transmitter has instantaneous channel state
information (CSI) available, it can compensate for a significant portion of
these outages by utilizing power allocation. In a conventional dual-hop
bidirectional amplify-and-forward (AF) relaying system, the relay already
has instantaneous CSI of both links available, as this is required for relay
gain adjustment. We therefore develop an optimal power allocation strategy
for the relay, which adjusts its instantaneous output power to the minimum
level required to avoid outages, but only if the required output power is
below some cutoff level; otherwise, the relay is silent in order to conserve
power and prolong its lifetime. The proposed scheme is proven to minimize
the system outage probability, subject to an average power constraint at
the relay and fixed output powers at the end nodes.
\end{abstract}

\section{Introduction}

Bidirectional (two-way) relaying has higher bandwidth efficiency compared to
unidirectional relaying \cite{R1}-\cite{R3}, and it is also a more suitable
option for applications where the end nodes intend to exchange information
(e.g., in interactive applications). The amplify-and-forward (AF) version of
bidirectional relaying, often referred to as the \textit{analog network
coding} \cite{R3}, has recently been extensively studied \cite{R4}-\cite{R10}.
In this paper, we consider bidirectional analog network coding
with fixed information rates, which is suitable for delay-sensitive applications, such
as bidirectional interactive speech and/or video communication. Fixed-rate
communication systems are usually characterized by the capacity outage
probability (OP), which is a relevant performance measure in quasi-static
(i.e., slowly fading) channels. In these systems, each transmitted codeword
is affected by only one channel realization.

Papers \cite{R4}, \cite{R5}, and \cite{R6} focus on the OP analysis of dual-hop bidirectional
AF systems. However, \cite{R4} and \cite{R5} assume that the end nodes and the relay
transmit with fixed powers, although they have instantaneous channel state
information (CSI) available. In particular, for each channel realization, the relay needs
CSI to set the amplification gain factor, whereas the end nodes need CSI to
cancel out self-interference and to decode the desired codeword. The
available CSI in fixed-rate bidirectional relaying systems can be used for
optimal power allocation at the end nodes and the relay. Papers \cite{R6}-%
\cite{R9} develop power allocation schemes that optimize different system
objectives; \cite{R6} minimizes the OP of either one of the two traffic
flows, \cite{R8} proposes a power allocation that balances the individual
outage probabilities of the two end nodes, \cite{R7} maximizes the
instantaneous sum rate of bidirectional AF systems, and \cite{R9} minimizes the
total consumed energy such that the OPs of both traffic flows are maintained
below some predefined values. However, the power allocations in these papers
are subject to \textit{short-term power} constraints, which limits the
codeword power in each channel realization.

On the other hand, it is also possible to adopt \textit{average (long-term)
power} constraints so as to limit the average power of all codewords over
all channel realizations \cite{R13}. For point-to-point channels, such power
adaptation is known as \textit{truncated channel inversion} and has been
introduced in \cite{R12}. For unidirectional relaying, optimal power
allocation for source and relay has been studied for both conventional
amplify-and-forward \cite{R10} and decode-and-forward (DF) \cite{R11}
relaying systems under various average power constraints. Optimal power
allocation has been shown to introduce significant performance improvement
relative to constant power transmission \cite{R10}-\cite{R12}. However, the
literature does not offer similar results for bidirectional relaying with
long-term power constraints.

In this work, we derive power control strategies for the relay in
bidirectional dual-hop AF relaying systems. For predefined constant rates in
both directions, the proposed power allocation achieves minimization of the
system OP assuming an average power constraint at the relay. The end nodes
are assumed to be simple nodes, equipped  with cheap power amplifiers,
and therefore are unable to support the high peak-to-average power ratios
at their output, required for channel inversion. Thus, the end nodes transmit
with fixed powers.

The relay applies the proposed optimal power control strategy based on the
already available knowledge of the channel coefficients of both links. Intuitively, it is not
necessary for the relay to transmit at its maximum available power in each
transmission cycle, but to transmit with the minimum power required to avoid
outages, or sometimes even be silent when outages are unavoidable, thus
conserving power. In other words, we allow outages to occur in cases of deep
fades, but for the rest of the time we ensure successful transmissions at
the predefined constant transmission rate.

\section{System and channel model}

The considered bidirectional relaying system comprises two end-nodes ($%
S_{1}$ and $S_{2}$) and a half-duplex AF relay $R$. The bidirectional
communication consists of two parallel unidirectional communication
sessions, $S_{1}\rightarrow S_{2}$ and $S_{2}\rightarrow S_{1}$. Each
communication session is realized at a fixed information rate, $R_{01}$ and $%
R_{02}$, respectively. The OP for this system is defined as the probability
that at least one (or both) of the communications sessions is in outage.

The complex coefficients of the $S_{1}-R$ and $S_{2}-R$ channels are denoted
by $\alpha $ and $\beta $, respectively, whereas their respective squared amplitudes are $%
x=|\alpha |^{2}$ and $y=|\beta |^{2}$ with average values $\Omega _{X}$ and $%
\Omega _{Y}$. We assume that the $R-S_{1}$ channel is
reciprocal to the $S_{1}-R$ channel, and the $R-S_{2}$ channel is reciprocal
to the $S_{2}-R$ channel. We adopt the Rayleigh block fading model, which
means that the values of $\alpha $ and $\beta $ are constant for each
transmission cycle, but change from one transmission cycle to the next. In
each transmission cycle, the pair $\left( x,y\right) $ denotes the channel
state, where both $x$ and $y$ follow an exponential probability density
function (PDF). The received signals at the end nodes and the relay are corrupted
by additive white Gaussian noise (AWGN) with zero mean and unit variance.

We assume a two-phase transmission cycle consisting of the mutiple access
phase and the broadcast phase. In two-phase relaying, bidirectional
communication can be realized only via the relayed link $S_{1}-R-S_{2}$. In
the first phase, $S_{1}$ and $S_{2}$ simultaneously transmit their
codewords $s_{1}(t)$ and $s_{2}\left( t\right) $ with respective information
rates $R_{01}$ and $R_{02}$ and respective fixed output powers $P_{S1}$ and $%
P_{S2}$. The AF relay receives the composite signal and amplifies it by applying the
following gain:
\begin{equation}
G^{2}=\frac{P_{R}\left( x,y\right) }{P_{S1}x+P_{S2}y+1}\text{,}  \label{1}
\end{equation}%
where $P_{R}\left( x,y\right) $ is the relay output power, adjusted
according to the channel state $\left( x,y\right) $. The relay is assumed to
know $x$ and $y$.

In the second phase, the relay broadcasts the composite signal and the end
nodes receive it. Node $S_{1}$ is assumed to know the products $G\alpha
\beta $ and $G\alpha ^{2}$, whereas node $S_{2}$ is assumed to know the
products $G\alpha \beta $ and $G\beta ^{2}$. Using the available CSI, each
end node ``subtracts" its own signal from the received one, and then attempts
decoding of the faded and noisy version of the desired signal that
originates from the other end node. The desired signal-to-noise ratio (SNR)
at $S_{2}$ is thus given by
\begin{equation}
\gamma _{2}=\frac{P_{S1}x\text{ }P_{R}y}{P_{S1}x+\left( P_{S2}+P_{R}\right)
y+1}\text{,}  \label{2}
\end{equation}%
whereas the SNR at $S_{1}$ is given by
\begin{equation}
\gamma _{1}=\frac{P_{S2}x\text{ }P_{R}y}{\left( P_{S1}+P_{R}\right)
x+P_{S2}y+1}\text{. }  \label{3}
\end{equation}%
Based on (\ref{2}) and (\ref{3}), the relaying system can support $S_{1}$'s
transmission rate, $R_{01}$, if the instantaneous capacity of the end-to-end
channel $S_{1}\rightarrow R\rightarrow S_{2}$ exceeds this rate such that $%
\frac{1}{2}\log _{2}\left( 1+\gamma _{2}\right) \geq R_{01}$, where the
pre-log factor $1/2$ is due to the two-phase transmission cycle, or
equivalently $\gamma _{2}\geq \delta _{1}$, where $\delta _{1}=2^{2R_{01}}-1$%
. Similarly, the relaying system can support $S_{2}$'s transmission rate, $%
R_{02}$, if the instantaneous capacity of the end-to-end channel $%
S_{2}\rightarrow R\rightarrow S_{1}$ exceeds this rate such that $\frac{1}{2}%
\log _{2}\left( 1+\gamma _{1}\right) \geq R_{02}$, or equivalently $\gamma
_{1}\geq \delta _{2}$, where $\delta _{2}=2^{2R_{02}}-1$. The system OP is
therefore determined as
\begin{eqnarray*}
P_{out} = \Pr \left\{ \frac{1}{2}\log _{2}\left( 1+\gamma _{2}\right) <R_{01}%
\text{ OR }\frac{1}{2}\log _{2}\left( 1+\gamma _{1}\right) <R_{02}\right\}
\end{eqnarray*}%
\begin{equation}
=1-\Pr \left\{ \gamma _{2}\geq \delta _{1}\text{ AND }\gamma _{1}\geq \delta
_{2}\right\} \text{, }  \label{4}
\end{equation}%
where $\gamma _{2}$ and $\gamma _{1}$ are given by (\ref{2}) and (\ref{3}),
respectively.

\section{Outage Minimization}

We now present the optimal power allocation (OPA) strategy at the relay, $%
P_{R}^{\ast }\left( x,y\right) $, that minimize the system OP, subject to
a (long-term) average power constraint at the relay, $P_{avg}$, and fixed
output powers at the end nodes, $P_{S1}$ and $P_{S2}$. The OPA is specified
in the following theorem.
\begin{theorem}
For the system OP defined as in (\ref{4}), the solution of the optimization
problem
\begin{equation*}
\underset{P_{R}\left( x,y\right) }{\text{minimize}}\text{ }P_{out}
\end{equation*}%
\begin{equation}
\text{subject to \ \ \ }E_{XY}\left[ P_{R}\left( x,y\right) \right] \leq
P_{avg}  \label{5}
\end{equation}%
where $E_{XY}\left[ \cdot \right] $ denotes expectation with respect to
random variables $X$ and $Y$, is given by
\end{theorem}
\begin{equation}
P_{R}^{\ast }\left( x,y,\rho \right) =\left\{
\begin{array}{c}
P_{R,st}\left( x,y\right) ,\text{ \ \ if }P_{R,st}\left( x,y\right) \leq \rho
\\
0,\text{ \ \ \ \ \ \ \ \ \ \ \ \ \ \ \ if }P_{R,st}\left( x,y\right) >\rho%
\end{array}%
\right.  \label{6}
\end{equation}%
where $P_{R,st}\left( x,y\right) $ is the minimum short-term relay power
that maintains zero OP, and $\rho $ is the cutoff threshold determined from
\begin{equation}
P_{avg}=E_{XY}\left[ P_{R,st}\left( x,y\right) \text{ }|P_{R,st}\left(
x,y\right) \leq \rho \right] \text{ .}  \label{7}
\end{equation}
\begin{proof}
The proof is presented in the Appendix.
\end{proof}
Intuitively, the solution in (\ref{6}) resembles the truncated channel
inversion in point-to-point communication links \cite{R12}. The minimum
possible short-term power $P_{R,st}\left( x,y\right)$ prevents
system outage events to the greatest possible extent, such that the average
relay output power is below the predefined value $P_{avg}$. The cutoff
threshold $\rho$ assures that the long-term average power constraint is
satisfied, such that the relay is silent if $P_{R,st}\left( x,y\right) $
exceeds $\rho $. In the following subsections, we determine $P_{R,st}\left(
x,y\right) $ and $\rho $.

\subsection{Minimum Short-term Power Required for Zero Outage}
The power control scheme that maintains zero outage probability and also
minimizes the relay output power is determined according to the following
theorem.
\begin{theorem}
The solution of optimization problem%
\begin{equation*}
\text{$\underset{P_{R}\left( x,y\right) }{\text{minimize}}$ }P_{R}\left(
x,y\right)
\end{equation*}%
\vspace{-5mm}
\begin{eqnarray}
\text{subject to \ \ }\frac{1}{2}\log _{2}\left( 1+\gamma _{2}\right) &\geq
&R_{01}  \notag \\
\frac{1}{2}\log {}_{2}\left( 1+\gamma _{1}\right) &\geq &R_{02}  \label{8}
\end{eqnarray}%
is given by
\begin{equation}
P_{R,st}\left( x,y\right) =\left\{
\begin{array}{c}
\max \left\{ \frac{\delta _{1}\left( 1+P_{S1}x+P_{S2}y\right) }{y\left(
P_{S1}x-\delta _{1}\right) },\frac{\delta _{2}\left(
1+P_{S1}x+P_{S2}y\right) }{x\left( P_{S2}y-\delta _{2}\right) }\right\}
\text{, } \\
\text{if }x\geq \frac{\delta _{1}}{P_{S1}}\text{ and }y\geq \frac{\delta _{2}%
}{P_{S2}} \\
\\
0,\text{ otherwise}%
\end{array}%
\right. .  \label{9}
\end{equation}
\end{theorem}
\begin{proof}
Optimization problem (\ref{8}) is a standard \textit{linear programming}
(LP) problem, whose solution is feasible because the intersection of the
constraints in (\ref{8}) is a non-empty set, and the solution lies at the
boundary of the intersection. Considering (\ref{2}), the first constraint
in (\ref{8}) is satisfied if the output powers of the relay and
the end node $S_{1}$ satisfy the following conditions:
\begin{equation}
P_{R}\geq \frac{\delta _{1}\left( 1+P_{S1}x+P_{S2}y\right) }{y\left(
P_{S1}x-\delta _{1}\right) }\text{ \ and \ }P_{S1}x-\delta _{1}\geq 0\text{.}
\label{10}
\end{equation}
The second constraint in (\ref{8}) is satisfied if the output powers of the
relay and the end node $S_{2}$ satisfy the following conditions:
\begin{equation}
P_{R}\geq \frac{\delta _{2}\left( 1+P_{S1}x+P_{S2}y\right) }{x\left(
P_{S2}y-\delta _{2}\right) }\text{ \ and \ }P_{S2}y-\delta _{2}\geq 0\text{.}
\label{11}
\end{equation}
Then, the intersection of (\ref{10}) and (\ref{11}) is given by (\ref{9}).
If either one of the conditions $P_{S1}x-\delta _{1}\geq 0$ and $%
P_{S2}y-\delta _{2}\geq 0$ is not satisfied, then the intersection is an
empty set, which practically means that the relay should be silent, i.e., $%
P_{R,st}\left( x,y\right) =0$. In this case, an outage event is unavoidable
regardless of the available short-term power at the relay. This concludes
the proof.
\end{proof}

\begin{figure}[tbp]
\centering
\includegraphics[width=3.5in]{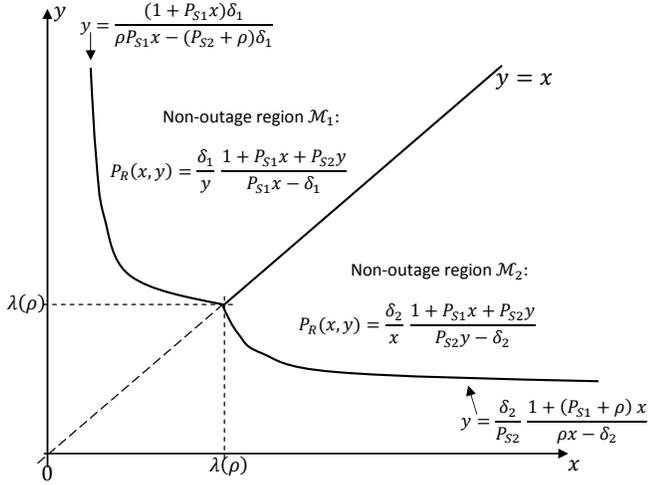} \vspace{-5mm} \vspace{-3mm}
\caption{Non-outage region $\mathcal{M} = \mathcal{M}_{1} \cup \mathcal{M}%
_{2}$.}
\label{fig1}
\end{figure}

\subsection{Average Relay Output Power and Cutoff Threshold $\protect\rho $}
The combination of (\ref{6}) and (\ref{9}) is the general analytical
solution to the considered outage minimization problem (\ref{5}). In order
to be able to obtain $\rho $ analytically, it is necessary to derive an
analytical expression for the average relay output power, $\overline{P}%
_{R}=E_{XY}\left[ P_{R,st}\left( x,y\right) \text{ }|P_{R,st}\left(
x,y\right) \leq \rho \right] $, where $P_{R,st}\left( x,y\right) $ is given
by (\ref{9}). Although possible, arriving at a general analytic expression for
$\overline{P}_{R}$ for arbitrary fixed output powers at the end nodes and
arbitrary rates requires lengthy derivations due to the complicated system
non-outage regions. Thus, due to the space restriction, here we focus on
the following special case:
\begin{equation}
\frac{P_{S1}}{\delta _{1}}=\frac{P_{S2}}{\delta _{2}}\text{.}  \label{12}
\end{equation}
Assuming (\ref{12}), the system non-outage region $\mathcal{M}=\left\{
\left( x,y\right) \left\vert P_{S1}x\geq \delta _{1},P_{S2}y\geq \delta
_{2},P_{R,st}\left( x,y\right) \leq \rho \right. \right\} $ is divided into
two non-overlapping regions, $\mathcal{M}=\mathcal{M}_{1}\cup \mathcal{M}_{2}$,
such that
\begin{equation}
\mathcal{M}_{1}:\frac{\delta _{1}\left[ 1+\left( P_{S2}+\rho \right) y\right]
}{P_{S1}\left( \rho y-\delta _{1}\right) }\leq x\leq y\text{ and }y\geq
\lambda \left( \rho \right)  \label{13}
\end{equation}%
\begin{equation}
\mathcal{M}_{2}:\frac{\delta _{2}\left[ 1+\left( P_{S1}+\rho \right) x\right]
}{P_{S2}\left( \rho x-\delta _{2}\right) }\leq y\leq x\text{ and }x\geq
\lambda \left( \rho \right)  \label{14}
\end{equation}%
where
\begin{eqnarray}
\lambda \left( \rho \right) &=&\frac{\delta _{2}\left( P_{S1}+P_{S2}+\rho
\right) }{2P_{S2}\rho }  \notag \\
&&\times \left( 1+\sqrt{1+\frac{4P_{S2}\rho }{\delta _{2}\left(
P_{S1}+P_{S2}+\rho \right) ^{2}}}\right) \text{.}  \label{15}
\end{eqnarray}%
Note that $\lambda \left( \rho \right) >P_{S1}/\delta _{1}$. Fig. 1
graphically illustrates the non-outage regions $\mathcal{M}_{1}$ and $%
\mathcal{M}_{2}$. Considering (\ref{13})-(\ref{15}), the average output
power of the power control scheme (\ref{9}) is determined as
\begin{eqnarray}
\overline{P}_{R} &=&E_{XY}\left[ P_{R,st}\left( x,y\right) \text{ }%
|P_{R,st}\left( x,y\right) \leq \rho \right]  \notag \\
&=&\int_{\lambda \left( \rho \right) }^{\infty }\int_{\frac{\delta _{1}\left[
1+\left( P_{S2}+\rho \right) y\right] }{P_{S1}\left( \rho y-\delta
_{1}\right) }}^{y}dxdyf_{X}\left( x\right) f_{Y}\left( y\right)  \notag \\
&&\times \frac{\delta _{1}\left( 1+P_{S1}x+P_{S2}y\right) }{y\left(
P_{S1}x-\delta _{1}\right) }  \notag \\
&&+\int_{\lambda \left( \rho \right) }^{\infty }\int_{\frac{\delta _{2}\left[
1+\left( P_{S1}+\rho \right) x\right] }{P_{S2}\left( \rho x-\delta
_{2}\right) }}^{x}dydxf_{X}\left( x\right) f_{Y}\left( y\right)  \notag \\
&&\times \frac{\delta _{2}\left( 1+P_{S1}x+P_{S2}y\right) }{x\left(
P_{S2}y-\delta _{2}\right) }\text{.}  \label{16}
\end{eqnarray}
The inner integral in (\ref{16}) can be solved in closed-form using the
exponential integral function $E_{1}\left( \cdot \right) $ \cite{R14}, which
is omitted here due to space limitation. Therefore, (\ref{16}) can be
expressed as a single integral, which depends on the cutoff threshold $\rho$.
Thus, $\rho $ can be determined numerically for a given $P_{avg}$.

\subsection{Optimum Power Allocation and Minimum Outage Probability}
Exploiting (\ref{13})-(\ref{15}), we combine (\ref{6}) and (\ref{9}) to
arrive at the final OPA expression
\begin{equation} \label{17}
P_{R}^{\ast }\left( x,y,\rho \right) = \qquad \qquad \qquad\qquad\qquad\qquad\qquad\qquad \,
\end{equation}
\vspace{-5mm}
\begin{eqnarray*}
\left\{
\begin{array}{c}
\frac{\delta _{1}\left( 1+P_{S1}x+P_{S2}y\right) }{y\left( P_{S1}x-\delta
_{1}\right) }\text{, if }\frac{\delta _{1}\left[ 1+\left( P_{S2}+\rho
\right) y\right] }{P_{S1}\left( \rho y-\delta _{1}\right) }\leq x\leq y\text{
and }y\geq \lambda \left( \rho \right) \\
\frac{\delta _{2}\left( 1+P_{S1}x+P_{S2}y\right) }{x\left( P_{S2}y-\delta
_{2}\right) }\text{, if }\frac{\delta _{2}\left[ 1+\left( P_{S1}+\rho
\right) x\right] }{P_{S2}\left( \rho x-\delta _{2}\right) }\leq y\leq x\text{
and }x\geq \lambda \left( \rho \right) \\
0\qquad \qquad \qquad \qquad ,\text{ otherwise}\qquad \qquad \qquad \qquad
\qquad
\end{array}
\right.
\end{eqnarray*}
Assuming power allocation (\ref{17}), the system OP\ (\ref{4}) is determined
as $P_{out}=1-\Pr \left\{ (x,y)\in \mathcal{M}\right\} $, i.e.,

\begin{eqnarray*}
P_{out}\left( \rho \right) &=&1-\int_{\lambda \left( \rho \right) }^{\infty
}\int_{\frac{\delta _{1}\left[ 1+\left( P_{S2}+\rho \right) y\right] }{%
P_{S1}\left( \rho y-\delta _{1}\right) }}^{y}f_{X}\left( x\right)
f_{Y}\left( y\right) dxdy\qquad \, \\
&&-\int_{\lambda \left( \rho \right) }^{\infty }\int_{\frac{\delta _{2}\left[
1+\left( P_{S1}+\rho \right) x\right] }{P_{S2}\left( \rho x-\delta
_{2}\right) }}^{x}f_{X}\left( x\right) f_{Y}\left( y\right) dydx\qquad
\qquad \qquad \\
\end{eqnarray*}
\vspace{-10mm}
\begin{eqnarray}
=1+e^{-\lambda \left( \frac{1}{\Omega _{X}}+\frac{1}{\Omega _{Y}}\right)
}\qquad \qquad \qquad \qquad \qquad \qquad \qquad \quad  \notag \\
-\frac{1}{\Omega _{Y}}\int_{\lambda }^{\infty }\exp \left[ -\left( \frac{z}{%
\Omega _{Y}}+\frac{\delta _{1}}{\Omega _{X}P_{S1}}\frac{1+\left( P_{S2}+\rho
\right) }{\rho z-\delta _{1}}\right) \right] dz  \notag \\
-\frac{1}{\Omega _{X}}\int_{\lambda }^{\infty }\exp \left[ -\left( \frac{z}{%
\Omega _{X}}+\frac{\delta _{2}}{\Omega _{Y}P_{S2}}\frac{1+\left( P_{S1}+\rho
\right) }{\rho z-\delta _{2}}\right) \right] dz\text{.}  \label{18}
\end{eqnarray}
For a given $P_{avg}$, the cutoff threshold $\rho$ is determined from
(\ref{16}), and (\ref{18}) leads to the minimum OP of the considered
bidirectional relaying system. However, note that (\ref{18}) is also the
system OP for a relay with constant output power set as $P_{R}\left(
x,y,\rho \right) =\rho $. For $\rho \rightarrow \infty $ (i.e.,
$P_{avg}\rightarrow \infty $), (\ref{18}) attains its minimum value given by
\begin{equation}
P_{out}^{\min }=1-e^{-\frac{\delta _{1}}{P_{S1}}\left( \frac{1}{\Omega_{X}}+
\frac{1}{\Omega _{Y}}\right) }\text{. }  \label{19}
\end{equation}
Namely, regardless of the available relay power, the outages are imminent if
the channel between the originating end node and the relay cannot support the
desired rate. We again note that (\ref{16})-(\ref{19}) are valid under assumption
(\ref{12}).

\subsection{Minimization of Average Relay Power}
The following theorem determines the solution of the dual optimization
problem of (\ref{5}), which minimizes the average relay output power subject
to some target system OP, $P_{out}^{\max}$.
\begin{theorem}
For system OP $P_{out}$ defined as per (\ref{4}), the solution of
optimization problem
\begin{equation*}
\text{$\underset{P_{R}\left( x,y\right) }{\text{minimize}}$ \ }E_{XY}\left[
P_{R}\left( x,y\right) \right]
\end{equation*}%
\vspace{-5mm}
\begin{equation}
\text{subject to \ \ }P_{out}\leq P_{out}^{\text{target}}\text{,}  \label{20}
\end{equation}
is given by \vspace{-3mm}
\begin{equation}
P_{R}^{\ast \ast }\left( x,y,\mu \right) =\left\{
\begin{array}{c}
P_{R,st}\left( x,y\right) ,\text{ \ \ if }P_{R,st}\left( x,y\right) \leq \mu
\\
0,\text{ \ \ \ \ \ \ \ \ \ \ \ \ \ \ \ if }P_{R,st}\left( x,y\right) >\mu%
\end{array}
\right. ,  \label{21}
\end{equation}
where $P_{R,st}\left( x,y\right) $ is given by (\ref{9}), whereas the cutoff
threshold $\mu $ satisfies
\begin{equation}
P_{out}\left( \mu \right) =P_{out}^{\text{target}}\text{,}  \label{22}
\end{equation}
where $P_{out}\left( \mu \right) $ given by (\ref{18}).
\end{theorem}
\begin{proof}
A sketch of the proof is provided in Appendix B.
\end{proof}

\begin{figure}[tbp]
\centering
\includegraphics[width=3.5in]{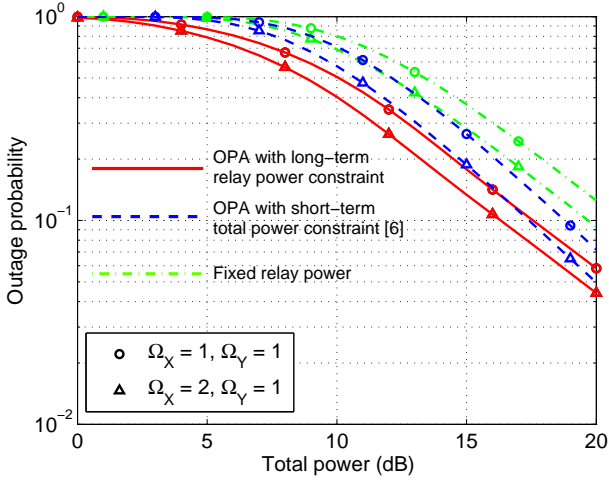} \vspace{-7mm}
\caption{Outage performance improvement due to power allocation.} \vspace{-4mm}
\label{fig3}
\end{figure}
\begin{figure}[tbp]
\centering
\includegraphics[width=3.5in]{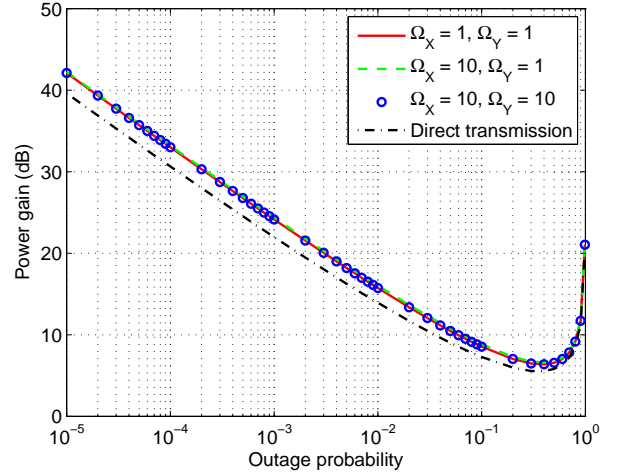} \vspace{-7mm}
\caption{Power savings at the end nodes and the relay.} \vspace{-4mm}
\label{fig4}
\end{figure}

\section{Numerical Examples}
We illustrate the performance gains of the proposed OPA for the considered
three-node AF relaying system under two different scenarios. For both
scenarios, the rates are fixed to $R_{01}=R_{02}=1/2$.

\textbf{Scenario 1}: We compare the system OP of the proposed OPA with:
(\textit{i}) fixed power allocation (FPA), and (\textit{ii}) OPA with
short-term total power constraint as derived in \cite{R6}, denoted as short-term OPA.
Note, for each channel realization, the short-term OPA optimally allocates
the same amount of the total available power to the relay and the two end nodes.
For a given total available power $P_{T}$, the system using the proposed OPA employs $%
P_{S1}=P_{S2}=P_{avg}=P_{T}/3$; the system using FPA employs $%
P_{S1}=P_{S2}=P_{R}=P_{T}/3$; the system using the short-term
OPA employs [6, Es. (15)-(17)]: $P_{S1}=0.5P_{T}\sqrt{y}/\left( \sqrt{x}+%
\sqrt{y}\right) $, $P_{S2}=0.5P_{T}\sqrt{x}/\left( \sqrt{x}+\sqrt{y}\right)$,
and $P_{R}=0.5P_{T}$.

According to Fig. 2, the proposed OPA leads to a significant OP improvement
relative to FPA for any given $P_{T}$. In each coding block, the OPA
scheme allocates just enough power to the relay so as to maintain the
desired rates, and the relay is silent when ``deep fades" occur. On the other
hand, FPA always spends the same power in each coding block regardless of
the channel state. The proposed OPA also performs better than the short-term
OPA in the considered $P_{T}$ region, although the OP gain diminishes with
increasing $P_{T}$. Clearly, the short-term OPA outperforms the proposed OPA
above a certain $P_{T}$, because the short-term OPA employs a global short-term
power constraint, i.e., the sums of all node powers is constrained, whereas we employ
individual power constraints and the end nodes transmit with fixed power.

\textbf{Scenario 2}: For a given system OP, we consider the \textit{power
gain at the relay} for OPA $P_{R}^{\ast \ast }\left( x,y,\mu \right) $ as
per (\ref{21}), relative to FPA $P_{R\text{,fixed}}$, such that the power
gain is defined as $P_{R\text{,fixed}}/E_{XY}\left[ P_{R}^{\ast \ast }\left(
x,y,\mu \right) \right] $. Both OPA and FPA\ lead to the same OP ($P_{out}^{%
\text{target}}$) by setting $P_{R\text{,fixed}}=\mu $ such that $\mu $ is
determined from (\ref{22}). We also set $P_{S1}=P_{S2}=\mu $. According to
Fig. 3, the power gains are remarkably high when the OP is low, because
channel inversion is applied to almost all channel states ($\rho $ has high
value). For relatively high OPs (OP between 0.3 and 0.7), the power gain is
minimized (but is still above 5 dB), because the nodes are often silent
although channel states are not exposed to ``deep fades". For comparison
purposes, the dotted line in Fig. 3 denotes the power gain for truncated
channel inversion over a point-to-point communication link in Rayleigh
fading, which can be shown to be \cite{R12}
\begin{equation}
\frac{1}{-\log \left( 1-P_{out}^{\text{target}}\right) \cdot E_{1}\left(
-\log \left( 1-P_{out}^{\text{target}}\right) \right) }\text{.}  \label{23}
\end{equation}

\section{Conclusion}
In this paper, we determined the optimal power allocation at the relay
that minimizes the OP of a conventional dual-hop bidirectional AF relaying
system with fixed rates, subject to a long-term power constraint at the relay.
The end nodes are assumed to be simple communication devices that cannot
adapt their output powers. The general solution resembles the truncated channel inversion
scheme for point-to-point links. For a special case, we have also derived analytical
expressions that allow the evaluation of the cutoff threshold and the OP. The proposed scheme
achieves remarkable performance improvements and/or power savings. Most importantly,
these benefits come without additional cost for the system, because the required CSI has to be
acquired by the relay anyways for adjusting its amplification.

\appendices
%\section{}
%\renewcommand{\theequation}{\thesection.\arabic{equation}}
%\setcounter{equation}{0}

\section{Proof of Theorem 1}

This proof is inspired by \cite{R13}. Theorem 1 is true if the following two propositions are true:
\textbf{Proposition 1}: If $P_{R}^{\ast }(x,y)>0$, then $P_{R}^{\ast}(x,y)=P_{R,st}(x,y)$,
\textbf{Proposition 2}: $P_{R}^{\ast }(x,y)>0$ if and only if $P_{R,st}(x,y)\leq \rho $.

\textbf{Proof of Proposition 1}: We prove Proposition 1 by contradiction.
Assume that Proposition 1 is not true and that $P_{R}^{+}(x,y)$ is the
optimal power. Then,  $P_{R}^{+}(x,y)$ differs from $P_{R}^{*}(x,y)$ if for some or all $(x,y)$, $P_{R}^{+}(x,y)$ is such that $P_{R}^{+}(x,y)>0$ and $%
P_{R}^{+}(x,y)\neq P_{R,st}(x,y)$. Let $\mathcal{S}$ represent the set of $%
(x,y)$ points for which $P_{R}^{+}(x,y)>0$ and $P_{R}^{+}(x,y)\neq
P_{R,st}(x,y)$. Then, for the points $(x,y)\in \mathcal{S}$, there are three
possibilities for the values of $P_{R}^{+}(x,y)$. Either $%
P_{R}^{+}(x,y)<P_{R,st}(x,y)$ for all $(x,y)\in \mathcal{S}$, or $%
P_{R}^{+}(x,y)>P_{R,st}(x,y)$ for all $(x,y)\in \mathcal{S}$, or the set $%
\mathcal{S}$ is comprised of two sets $\mathcal{S}_{1}$ and $\mathcal{S}_{2}$%
, where for $(x,y)\in \mathcal{S}_{1}$, $P_{R}^{+}(x,y)<P_{R,st}(x,y)$ and
for $(x,y)\in \mathcal{S}_{2}$, $P_{R}^{+}(x,y)>P_{R,st}(x,y)$. We will
prove that the first two possibilities are not possible and, as a result,
the third possibility is also not possible, therefore $%
P_{R}^{+}(x,y)=P_{R,st}(x,y)=P_{R}^{\ast }(x,y)$ for all $(x,y)$ for which $%
P_{R}^{+}(x,y)>0$.

Assume the first possibility, i.e., for $(x,y)\in \mathcal{S}$, $%
P_{R}^{+}(x,y)<P_{R,st}(x,y)$. Then, $P_{R}^{+}(x,y)$ cannot be the optimal
solution since for the considered $(x,y)\in \mathcal{S}$ there is an outage
and therefore the outage probability will not change if $P_{R}^{+}(x,y)$ is
set to zero. Now assume the second possibility, i.e., for $(x,y)\in \mathcal{%
S}$, $P_{R}^{+}(x,y)>P_{R,st}(x,y)$. Then, $P_{R}^{+}(x,y)$ again cannot be
the optimal power solution since for $(x,y)\in \mathcal{S}$ there is no
outage and by setting $P_{R}^{+}(x,y)=P_{R,st}(x,y)$ the outage probability
will not change. Finally, the third possibility is a combination of the
first two and therefore cannot be true. Hence, Proposition 1 is true.

\textbf{Proof of Proposition 2}: Again, we prove Proposition 2 by
contradiction. Let $\mathcal{D}$ represent the set of points for which $%
P_{R,st}(x,y)\leq \rho $. Then, for the proposed optimal solution, given by (%
\ref{6}), the average power is given by
\begin{eqnarray}
E_{XY}\{P_{R}^{\ast }(x,y)\} =\int_{x=0}^{\infty
}\int_{y=0}^{\infty }P_{R}^{\ast }(x,y)f_{XY}(x,y)dxdy
\\
\hspace{-14mm}=\int \!\!\!\!\int_{(x,y)\in \mathcal{D}}\hspace{-6mm}
P_{R,st}(x,y)f_{XY}(x,y)dxdy=P_{avg}.\label{eq_2}
\end{eqnarray}
On the other hand, the outage probability is given by
\begin{equation} \label{eq_2a}
P_{out}^{\ast }=1-\int \!\!\!\!\int_{(x,y)\in \mathcal{D}}\hspace{-6mm}%
f_{XY}(x,y)dxdy,
\end{equation}
i.e., there is no outage for those $(x,y)$ for which the relay's power
$P_{R}(x,y)\geq P_{R,st}$, and this is satisfied only  for $(x,y)\in \mathcal{D}$, hence
comes (\ref{eq_2a}).

Now, let us assume that Proposition 2 is not true and that $P_{R}^{+}(x,y)$
is the optimal power. Since Proposition 1 is true, it follows that if the optimal power    $P_{R}^{+}(x,y)$ is nonzero for points $(x,y)$,  $P_{R}^{+}(x,y)=P_{R,st}(x,y)$ must hold. However, this is also true for $P_{R}^{*}(x,y)$, i.e., $P_{R}^{*}(x,y)=P_{R,st}(x,y)$ for $(x,y)\in \mathcal{D}$ and $P_{R}^{*}(x,y)=0$ for $(x,y)\notin \mathcal{D}$.
Hence, $P_{R}^{+}(x,y)$ will differ from $P_{R}^{\ast
}(x,y) $ if and only if for some or all points $(x,y)\notin \mathcal{D}$, $%
P_{R}^{+}(x,y)>0$ holds. Moreover, according to Preposition 1, since $P_{R}^{+}(x,y)>0$ then $P_{R}^{+}(x,y)$ must be $P_{R}^{+}(x,y)=P_{R,st}(x,y)$. Hence, $P_{R}^{+}(x,y)$ will differ from $P_{R}^{\ast
}(x,y) $ if for some (or all) $(x,y)\notin \mathcal{D}$, $P_{R}^{+}(x,y)=P_{R,st}(x,y)$ holds.

 Now, note that $P_{R,st}(x,y)\leq \rho $ if and only if $(x,y)\in \mathcal{D}$. As a result,  for $(x,y)\notin \mathcal{D}$, $P_{R,st}(x,y) $ must be $P_{R,st}(x,y)>\rho $.  Therefore, let us put these points $%
(x,y)\notin \mathcal{D}$ for which $P_{R}^{+}(x,y)=P_{R,st}(x,y)>\rho $
holds in the set $\mathcal{S}$.  Now, since $%
P_{R}^{+}(x,y)$ has to satisfy the power constraint it follows that the
following must hold:
\begin{eqnarray}
&&\hspace{-6mm}E_{XY}\{P_{R}^{+}(x,y)\}=\int_{x=0}^{\infty
}\int_{y=0}^{\infty }P_{R}^{+}(x,y)f_{XY}(x,y)dxdy  \notag  \label{eq_1} \\
&&\hspace{-4mm}=P_{avg}=\int \!\!\!\!\int_{(x,y)\in \mathcal{S}}\hspace{-8mm}%
P_{R}^{+}(x,y)f_{XY}(x,y)dxdy  \notag \\
&&\hspace{+8mm}+\int \!\!\!\!\int_{(x,y)\in \mathcal{D}}\hspace{-8mm}%
P_{R}^{+}(x,y)f_{XY}(x,y)dxdy.
\end{eqnarray}%
Since for $(x,y)\in \mathcal{S}$, $P_{R}^{+}(x,y)=P_{R,st}(x,y)>\rho$ we can express and denote the    first integral in (\ref{eq_1}) as
\begin{equation} \label{eq_2_2}
\int \!\!\!\!\int_{(x,y)\in \mathcal{S}}\hspace{-6mm}%
P_{R,st}(x,y)f_{XY}(x,y)dxdy =\epsilon,
\end{equation}
where $\epsilon >0$ must hold. Otherwise, if $\epsilon =0$ then $\mathcal{S}$
is an empty set in which case $P_{R}^{+}(x,y)=P_{R}^{\ast }(x,y)$. Combining
(\ref{eq_1}) and (\ref{eq_2_2}), the second integral in (\ref{eq_1}) can be
written as
\begin{equation} \label{eq_11}
\int \!\!\!\!\int_{(x,y)\in \mathcal{D}}\hspace{-8mm}%
P_{R}^{+}(x,y)f_{XY}(x,y)dxdy=P_{avg}-\epsilon .
\end{equation}
However, since (\ref{eq_2}) holds, in order for (\ref{eq_11}) to hold, it follows that for some points $%
(x,y)\in \mathcal{D}$, $P_{R}^{+}(x,y)$ has to be nonequal to $P_{R,st}(x,y)$. However, according to Preposition 1,  for any $(x,y)$, for which  $P_{R}^{+}(x,y)\neq P_{R,st}(x,y)$, $P_{R}^{+}(x,y)=0$ must hold. Hence, in order for (\ref{eq_11}) to hold, for some points $%
(x,y)\in \mathcal{D}$, $P_{R}^{+}(x,y)=0$.  Let us put the points $%
(x,y)\in \mathcal{D}$ for which $P_{R}^{+}(x,y)=0$ into the set $\mathcal{D}%
_{0}$ and the rest of the points in  $\mathcal{D}$ for which
$P_{R}^{+}(x,y)=P_{R,st}(x,y)$  in the set $\mathcal{D}_{1}$.  Therefore, (%
\ref{eq_11}) can be written as
\begin{eqnarray}
&&\hspace{-6mm}\int \!\!\!\!\int_{(x,y)\in \mathcal{D}}\hspace{-8mm}%
P_{R}^{+}(x,y)f_{XY}(x,y)dxdy=\int \!\!\!\!\int_{(x,y)\in \mathcal{D}_{0}}%
\hspace{-8mm}0\times f_{XY}(x,y)dxdy  \notag  \label{eq_n1} \\
&&\hspace{-3mm}+\int \!\!\!\!\int_{(x,y)\in \mathcal{D}_{1}}\hspace{-8mm}%
P_{R,st}(x,y)f_{XY}(x,y)dxdy  \notag \\
&&\hspace{-6mm}=\int \!\!\!\!\int_{(x,y)\in \mathcal{D}_{1}}\hspace{-8mm}%
P_{R,st}(x,y)f_{XY}(x,y)dxdy=P_{avg}-\epsilon .
\end{eqnarray}%
On the other hand, using $\mathcal{D}=%
\mathcal{D}_{0}\cup \mathcal{D}_{1}$, (\ref{eq_2}) can also be written  as
\begin{eqnarray}
&&\hspace{-6mm}\int \!\!\!\!\int_{(x,y)\in \mathcal{D}}\hspace{-8mm}%
P_{R,st}(x,y)f_{XY}(x,y)dxdy  \notag  \label{eq_n2} \\
&&\hspace{-6mm}=\int \!\!\!\!\int_{(x,y)\in \mathcal{D}_{0}}\hspace{-8mm}%
P_{R,st}(x,y)f_{XY}(x,y)dxdy  \notag \\
&&\hspace{-6mm}+\int \!\!\!\!\int_{(x,y)\in \mathcal{D}_{1}}\hspace{-8mm}%
P_{R,st}(x,y)f_{XY}(x,y)dxdy=P_{avg} .
\end{eqnarray}%
Subtracting (\ref{eq_2_2}) from (\ref{eq_n2}), we obtain
\begin{eqnarray}
&&\hspace{-6mm}\int \!\!\!\!\int_{(x,y)\in \mathcal{D}_{0}}\hspace{-8mm}%
P_{R,st}(x,y)f_{XY}(x,y)dxdy  \notag  \label{eq_n3} \\
&&\hspace{-6mm}+\int \!\!\!\!\int_{(x,y)\in \mathcal{D}_{1}}\hspace{-8mm}%
P_{R,st}(x,y)f_{XY}(x,y)dxdy  \notag \\
&&\hspace{-6mm}-\int \!\!\!\!\int_{(x,y)\in \mathcal{S}}\hspace{-6mm}%
P_{R,st}(x,y)f_{XY}(x,y)dxdy=P_{avg}-\epsilon .
\end{eqnarray}%
Hence, we obtain the same right hand side in (\ref{eq_n3}) as in (\ref{eq_n1}). Therefore, we can equivalent the left hand sides of   (\ref{eq_n3}) and (\ref{eq_n1}), and after some manipulations  obtain
\begin{eqnarray}
&&\hspace{-6mm}\int \!\!\!\!\int_{(x,y)\in \mathcal{D}_{0}}\hspace{-8mm}%
P_{R,st}(x,y)f_{XY}(x,y)dxdy  \notag  \label{eq_n4} \\
&&\hspace{-6mm}=\int \!\!\!\!\int_{(x,y)\in \mathcal{S}}\hspace{-6mm}%
P_{R,st}(x,y)f_{XY}(x,y)dxdy .
\end{eqnarray}%
Since, $P_{R,st}(x,y)\leq \rho$  for $(x,y)\in \mathcal{D}_{0}$, and   $P_{R,st}(x,y)>\rho $
for $(x,y)\in \mathcal{S}$,  (\ref{eq_n4}) holds if
and only if
\begin{equation} \label{eq_n5}
\int \!\!\!\!\int_{(x,y)\in \mathcal{D}_{0}}\hspace{-8mm}f_{XY}(x,y)dxdy\geq
\int \!\!\!\!\int_{(x,y)\in \mathcal{S}}\hspace{-6mm}f_{XY}(x,y)dxdy
\end{equation}%
holds, where equality holds if and only if $\epsilon =0$, in which case both
$\mathcal{D}_{0}$ and $\mathcal{S}$ are empty sets, which means that $%
P_{R}^{+}(x,y)=P_{R}^{\ast }(x,y)$.

On the other hand, the outage probability obtained with $P_{R}^{+}(x,y)$ is
given by
\begin{equation}
P_{out}^{+}=1-\int \!\!\!\!\int_{(x,y)\in \mathcal{D}_{1}}\hspace{-6mm}%
f_{XY}(x,y)dxdy-\int \!\!\!\!\int_{(x,y)\in \mathcal{S}}\hspace{-6mm}%
f_{XY}(x,y)dxdy,  \label{eq_n6}
\end{equation}%
i.e., there is no outage for those $(x,y)$ for which the relay's power $%
P_{R}(x,y)\geq P_{R,st}$, and for $P_{R}^{+}(x,y)$ this holds only for $%
(x,y)\in \mathcal{D}_{1}$ and $(x,y)\in \mathcal{S}$, which leads to (\ref%
{eq_n6}). Inserting the bound in (\ref{eq_n5}) into (\ref{eq_n6}), we obtain
\begin{eqnarray}
&&\hspace{-6mm} P_{out}^{+}\geq 1-\int \!\!\!\!\int_{(x,y)\in \mathcal{D}_{1}}%
\hspace{-8mm}f_{XY}(x,y)dxdy-\int \!\!\!\!\int_{(x,y)\in \mathcal{D}}\hspace{%
-6mm}f_{XY}(x,y)dxdy  \notag  \label{eq_n7} \\
&&\hspace{-6mm}=1-\int \!\!\!\!\int_{(x,y)\in \mathcal{D}}\hspace{-8mm}%
f_{XY}(x,y)dxdy=P_{out}^{\ast } \,,
\end{eqnarray}%
where equality holds if and only if $\mathcal{S}$ in a empty set in which
case $P_{R}^{+}(x,y)=P_{R}^{\ast }(x,y)$. Hence, $P_{out}^{+}>P_{out}^{*}$ for any $P_{R}^{+}(x,y) \neq P_{R}^{*}(x,y)$. This concludes the proof.

\section{Proof of Theorem 3}
The solution to (\ref{20}) must satisfy Preposition 1 in Appendix A. One solution which satisfies Preposition 1  is $P_{R}^{\ast }(x,y)$ given by (\ref{21}). Let us denote any solution other than $P_{R}^{\ast }(x,y)$ that still satisfies Preposition 1 as $P_{R}^{+}(x,y)$. Then, following a similar procedure as in the proof of Preposition 2 in Appendix A, we can obtain the expressions for the system OP and average relay powers resulting from $P_{R}^{\ast}(x,y)$ and $P_{R}^{+}(x,y)$. Preposition 2 of Appendix A proves that, if the average values of $P_{R}^{\ast }(x,y)$ and $P_{R}^{+}(x,y)$ are equal, then the system OP resulting from solution $P_{R}^{\ast }(x,y)$ is less than the system OP resulting from solution $P_{R}^{+ }(x,y)$. Following a similar approach, if we set the system OPs resulting from the solutions $P_{R}^{\ast }(x,y)$ and $P_{R}^{+}(x,y)$ to be equal, then the average value of $P_{R}^{\ast }(x,y)$ is always less than the average value of $P_{R}^{+ }(x,y)$. Hence, $P_{R}^{\ast }(x,y)$ gives the minimal average relay output power for a given system OP. This completes the sketch of the proof.

\end{document}